\documentclass[conference]{IEEEtran}
\IEEEoverridecommandlockouts

\usepackage{cite}
\usepackage{amsmath,amssymb,amsfonts}
\usepackage{algorithmic}
\usepackage{graphicx}
\usepackage{subcaption}  
\usepackage{textcomp}
\usepackage{xcolor}
\usepackage{authblk}
\usepackage[linesnumbered,ruled,vlined]{algorithm2e}
\SetCommentSty{small}

\usepackage{multirow}
\usepackage{array}
\usepackage{arydshln}
\usepackage{makecell}
\usepackage{booktabs}

\usepackage[bookmarks=false]{hyperref}
\usepackage{flushend}
\usepackage{enumitem}
\usepackage{booktabs}  
\usepackage{float}
\usepackage{url}
\usepackage{amsthm}
\newtheorem{theorem}{Theorem}
\newtheorem{corollary}{Corollary}[theorem]

\usepackage[all=normal,paragraphs=normal,floats=normal,mathspacing=normal,wordspacing=tight,charwidths=tight,mathdisplays=normal,leading=tight]{savetrees}

\def\BibTeX{{\rm B\kern-.05em{\sc i\kern-.025em b}\kern-.08em
    T\kern-.1667em\lower.7ex\hbox{E}\kern-.125emX}}

\begin{document}

\title{Multi-dimensional hierarchical dictionary search \\ for large MIMO-OFDM systems\\
\thanks{This work is supported by the French national research agency (MoBAIWL, grant ANR-23-CE25-0013 and RIS3, grant ANR-23-CMAS-0023)}}

\author[1]{Nay Klaimi}
\author[1]{Philippe Mary}
\author[1]{Luc Le Magoarou}

\affil[1]{Univ Rennes, INSA Rennes, CNRS, IETR-UMR 6164, Rennes, France}

\newcommand{\remNK}[1]{{\scriptsize \color{red} [Nay: #1]}}

\maketitle

\begin{abstract}
Sparse recovery algorithms are of utmost importance for estimation processes in wireless communications. However, communication systems such as massive multiple input multiple output (MIMO) systems are rapidly growing in dimension, which consequently increases the computational complexity of these algorithms. This work proposes a low-complexity strategy for the efficient implementation of the ``atom selection step" in these greedy sparse recovery algorithms, based on the structural features of these systems. A theoretical justification is presented along with tests using realistic channel data, to demonstrate the computational gain induced by the proposed approach and compare it to the classical sparse recovery approach.
\end{abstract}


\section{Introduction}
While modern wireless networks strive for ever-higher throughput and reliability, they are hindered by the high computational cost of signal processing. For example massive multiple input multiple output (MIMO) systems \cite{Rusek2013,Larsson2014}, widely adopted to leverage large bandwidths, result in channels with hundreds to thousands of complex coefficients, making all signal processing tasks computationally intensive. In particular, sparse recovery algorithms that are now widely used for tasks such as channel estimation and beamforming \cite{Heath2016} rely on the multiplication of large matrices known as dictionaries. Consequently, as system architectures grow in scale and density, the demand for complexity-reduced variants of these algorithms increases.

In the context of large-scale systems like the MIMO Orthogonal Frequency Division Multiplexing (MIMO-OFDM) systems, specific features can be used to design low-complexity algorithms. First, the system is multi-dimensional; in the sense that it operates over several physical dimensions (several antennas at the receiver and transmitter, subcarriers), which gives its dictionary a Kronecker structure. Second, because the sensors (with uniform linear array (ULA) antennas at both ends) and subcarriers are uniformly distributed in their respective spaces, single-path channels, correspond to complex exponentials underlying the Fourier transform. In this work, these two structural properties are leveraged to propose a reduced-complexity method for the atom selection step, a standard operation in greedy sparse recovery.\\
{\noindent\bf Contributions.} This paper proposes a low-complexity multi-dimensional hierarchical atom selection to be used within greedy sparse recovery algorithms such as orthogonal matching pursuit (OMP)\cite{Pati1993}, specifically relevant for high-dimensional MIMO systems. While achieving the same results as the classical approach for noiseless measurements, and slightly degraded for noisy ones, this method substantially reduces the number of required mathematical operations typically of several orders of magnitude. More specifically, the main contributions are as follows:
\begin{itemize}[leftmargin=*, labelsep=0.5em] 
    \item An atom selection procedure that exploits the two aforementioned structural properties, the first is the Kronecker dictionary decomposition for the multi-dimensional aspect, along with the Fourier structure for the hierarchical atom search. The proposed method leverages the construction of flexible meta-atoms to avoid costly exhaustive search.
    \item A theoretical study justifying the building of the meta-atoms, relying on classical Fourier transform properties.
    \item An extensive set of experiments on realistic synthetic channels in a high-dimensional MIMO-OFDM system, validating the proposed method's computational efficiency and strong performance even when the classical approach becomes intractable due to system scale.
\end{itemize}  
{\noindent\bf Related work.} The proposed approach integrates and advances recent methods from the literature. The hierarchical search has been briefly introduced in \cite{Chatelier2023}, to accelerate dictionary learning for the matching pursuit algorithm. In contrast, this paper carries out an in-depth study of this method and provides further justifications. It offers a generic algorithm that could work for all sparse recovery algorithms and furthermore extends it to multi-dimensionality. Several works have addressed the acceleration of these algorithm in a similar spirit. A first line of research proposes a tree-based representation of the dictionary, where group representatives are computed via clustering, reducing the search complexity while maintaining reliable atom identification \cite{Jost2006,Ayremlou2014}. Another approach exploits the structure of the dictionary to achieve quasi-linear complexities\cite{Skretting2017}, while region-based search strategies have also been investigated \cite{Dorffer2019}. Moreover, the works \cite{palacios2022,Bayraktar2024,klaimi2026} introduce the multi-dimensional orthogonal matching pursuit (MOMP) algorithm. These studies use the kronecker structure of the dictionaries used for the OMP algorithm to reduce its complexity. The present work goes one step further in complexity reduction by combining MOMP with the new \emph{hierarchical search}, thus exploiting all available structures to reduce the complexity of the atom selection step. 
\section{System model}
\label{sec:syst}
A one-dimensional system is first considered for clarity in presenting the problem and the proposed solution strategy. This system may represent, for instance, a single input single output (SISO) OFDM system (frequency dimension only), a single-subcarrier MIMO or Massive MIMO system (space dimension only), etc.
In a general context, consider a one-dimensional system with $N$ sensors (subcarriers or antennas). Let $\boldsymbol{\gamma}=(\gamma_1,\dots,\gamma_N)^\mathsf{T}\in \mathbb{R}^N$ be the vector describing the positions of the sensors in the so-called observation domain. An \emph{atomic signal} expressed as $\mathbf{e}(u_k)=\mathsf{e}^{-\mathsf{j}2\pi\boldsymbol{\gamma} u_k}$ represents a single-path channel of the system, where $u_k$ is called the target parameter corresponding to the $k$-th path and belonging to the target domain $[u_{\mathsf{min}},u_{\mathsf{max}}]$.
A multi-path channel with $K$ propagation paths is a sparse linear combination of atomic signals, expressed as $\mathbf{h} = \sum_{k=1}^{K} \alpha_k \, \mathbf{e}(u_k),$
where $\alpha_k \in \mathbb{C}$ represents the complex gain of the $k$-th path.

Considering only the space dimension, the vector $\mathbf{p}=(p_1,\dots,p_N)^\mathsf{T}\in \mathbb{R}^N$ corresponds to the antenna positions and the target parameter will be the angle of departure (or the angle of arrival) $\theta$. Under the plane wave assumption, a single-path channel with a given direction corresponds to a steering vector (SV), which serves as the atomic signal in this case, expressed as
\begin{equation}
      \mathbf{e}(\theta_k) = \big(\mathrm{e}^{-\mathrm{j} 2 \pi \frac{p_1}{\lambda} \cos{\theta_k}} ,\ldots, \mathrm{e}^{-\mathrm{j} 2 \pi \frac{p_N}{\lambda} \cos{\theta_k}}\big)^\mathsf{T}\in \mathbb{C}^{N},
\end{equation}
for the $k$-th path, where $\lambda$ denotes the wavelength.

Considering only the frequency dimension, $\mathbf{f}=(f_1,\dots,f_N)^\mathsf{T}\in \mathbb{R}^N$ denotes the vector of subcarrier frequencies in the observation domain. The atomic signal, which corresponds to a frequency response vector (FRV) in this case, is
\begin{equation}
  \mathbf{e}(\tau_k) = \big(\mathrm{e}^{-\mathrm{j} 2 \pi f_1 \tau_k} ,\ldots, \mathrm{e}^{-\mathrm{j} 2 \pi f_{N} \tau_k}\big)^\mathsf{T} 
\in \mathbb{C}^{N},  
\end{equation}
corresponding to the $k$-th path with the propagation delay $\tau_k$ as a target parameter.

A SISO OFDM system with $N$ subcarriers in the uplink is adopted first as an arbitrary illustrative choice.
The multi-path channel can be expressed under this model as \begin{equation}
\label{eq:1D_multi-path_channel}
\mathbf{h} = \sum_{k=1}^{K} \alpha_l \, \mathbf{e}(\tau_k)\,.
\end{equation}

A more complex multi-dimensional system, better reflecting realistic scenarios in which the proposed approach becomes essential, will be used later for validation. It consists of a MIMO OFDM system operating over $N_S$ subcarriers, where the BS is equipped with a ULA with $N_B$ antennas and the users are also equipped with a ULA comprising $N_M$ antennas, yielding a number of elements $N=N_SN_BN_M$. The atomic signals of the three dimensions will respectively be named $\mathbf{e_S}$, $\mathbf{e_B}$ and $\mathbf{e_M}$.
The multi-path channel can be expressed under this model as \begin{equation}
\label{eq:3D_multi-path_channel}
\mathbf{h} = \sum_{k=1}^{K} \alpha_l \, \mathbf{e}_S(\tau_k)\otimes\mathbf{e}_B(\theta_k)\otimes\mathbf{e}_M(\vartheta_k)\,.
\end{equation}

Finally,  for both systems, we consider an uplink scenario in which the BS gets noisy measurements of the channels of the form 
\begin{equation}
    \mathbf{y}=\mathbf{h}+\mathbf{n}, 
\end{equation}
where $\mathbf{n}\sim \mathcal{CN}(0,\sigma^2)$ represents noise. The signal to noise ratio (SNR) can be computed as follows \[\text{SNR}\triangleq\frac{\|\mathbf{h}\|^2_2}{N\sigma^2}.\]
Such measurements can be obtained, for example, by transmitting orthogonal pilot sequences and correlating with them at the receiver; they will also be referred to as \emph{observations}.
\section{Problem Statement}
\label{sec:pb_statement}
The considered problem is recovering a signal $\mathbf{h} \in \mathbb{C}^N$  from a noisy measurement $\mathbf{y}$. The objective is either to denoise the received measurement (e.g. channel estimation) or to estimate physical parameters from it (e.g. AoA for localization). To do so, sparse recovery algorithms (e.g. MP \cite{Mallat1993} and OMP \cite{Pati1993}) rather consider estimation of the signal by a sparse approximation of the observation, in some dictionary defined as \( \mathbf{D}=\{\mathbf{a}_i\in \mathbb{C}^N\}^A_{i=1}  \),
where $\mathbf{a}$ is the so-called \emph{atom}, and $A$ the number of atoms in $\mathbf{D}$. The objective is to find a sparse representation $\boldsymbol{\beta}\in \mathbb{C}^A$ such that \(\mathbf{y}\approx\mathbf{D}\boldsymbol{\beta}\), with a sparsity level $s$; $\|\boldsymbol{\beta}\|_0=s$. In classical sparse recovery, the atoms take the form of atomic signals such that \begin{equation}
    \mathbf{D}=\{\mathbf{e}(\tau_i)\in \mathbb{C}^N\}^A_{i=1},
\end{equation}
where the $\tau_i$ values are evenly spaced over the target domain (delay domain here).
In order to reconstruct the signal, greedy sparse recovery algorithms search for the atoms of the dictionary that are most correlated with the signal and then project the observation onto the corresponding subspace, hence recovering it after several iterations. This critically relies on the atom selection step that is expressed as: \begin{equation}
\label{eq:opt_pb}
    \operatorname{Find} \mathbf{a}_\text{max} \in \arg\max_{\mathbf{a}_i}{|\langle \mathbf{a}_i, \boldsymbol{\epsilon} \rangle|}
\end{equation}
where $\boldsymbol{\epsilon}$ represents the residual of the algorithm initialized as $\boldsymbol{\epsilon}=\mathbf{y}$ and $\langle \mathbf{a}_i,\boldsymbol{\epsilon}\rangle=\mathbf{a}_i^\mathsf{H}\boldsymbol{\epsilon}$ represents the correlation (inner product) of the atom $\mathbf{a}_i$ with $\boldsymbol{\epsilon}$.

Let us define the \emph{response} of an atom to an atomic signal as their correlation, as a function of the target parameter $\tau$: \begin{equation}
    \mathbf{r}_i(\tau)=\mathbf{a}_i^\mathsf{H}\mathbf{e}(\tau)=\mathbf{e}(\tau_i)^\mathsf{H}\mathbf{e}(\tau).
\end{equation} 
Classically, the atom $\mathbf{a}_i$ will only have a non-negligible response when the delay parameter $\tau$ of the atomic signal is around $\tau_i$, as illustrated in Fig. \ref{fig:atom_response}.
\begin{figure}[H]
\centering
    \includegraphics[width=\columnwidth]{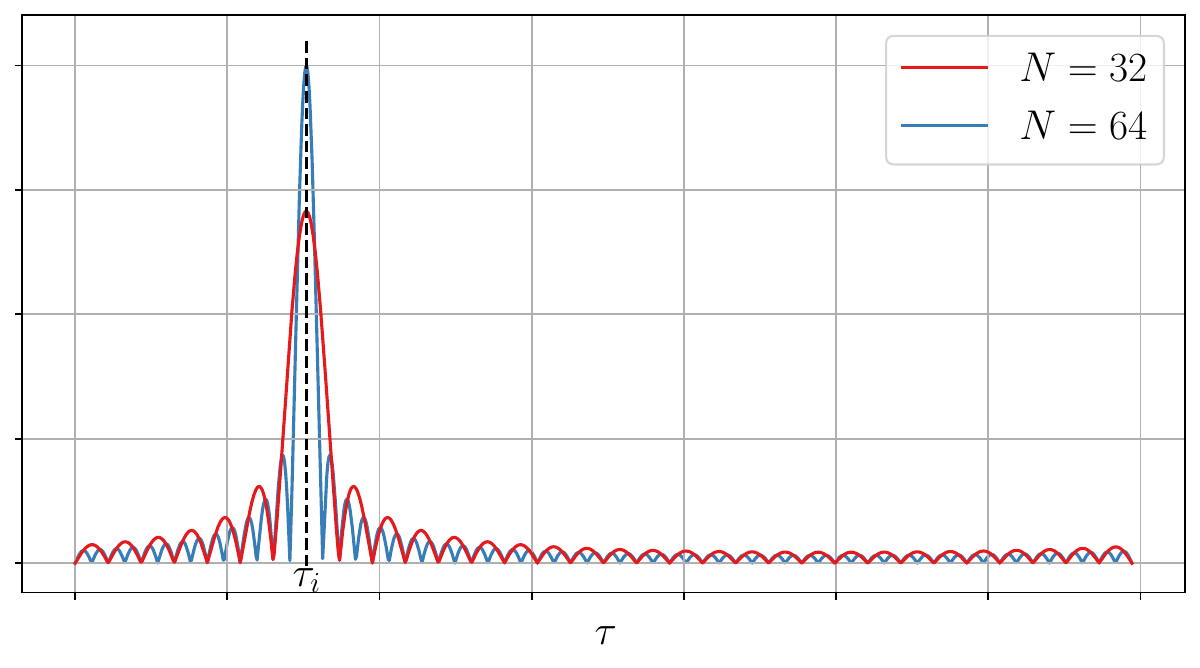}
    \caption{Amplitude of the response of the atom $\mathbf{e}(\tau_i)$ to atomic signals of the form $\mathbf{e}(\tau)$, as a function of $\tau$.}
    \label{fig:atom_response}
\end{figure}
Classical approaches thus require a high resolution dictionary to perform accurate atom selection (that is, a large number of atoms to cover all the target domain), as the response is strictly localized around each atom's delay. However, as the system dimensions grow, the response becomes more localized (as illustrated in Fig.\ref{fig:atom_response}) therefore, the number of atoms needed to achieve reliable sparse recovery becomes prohibitively large.

\section{Proposed Strategy}
\label{sec:prpsd_strategy}
This paper proposes a reduced complexity strategy for the atom selection step explained in Section \ref{sec:pb_statement}. To efficiently solve \eqref{eq:opt_pb}, the classical exhaustive search can be replaced with a \emph{hierarchical search}; iteratively searching the atom domain using rree-based search logic, while maintaining the same resolution. This search consists in introducing a new structure called \emph{meta-atom} that allows spreading the response of an atom over a controlled range of the target domain. These concepts will be explained in details in this section.
 
 \subsection{Meta-atoms contruction}
The meta-atom structure is obtained by modulating the classical atom $\mathbf{a}_i = \mathbf{e}(\tau_i)$ to spread its response. This is made possible by the Fourier property of the atomic signals, as established in Theorem~\ref{th:metaatoms} and Corollary~\ref{corr:metaatoms}.
 
\begin{theorem}
\label{th:metaatoms}
Computing the response of an atom \(\mathbf{a}\) to an atomic signal \( \mathbf{e}(u) \), is equivalent to computing the Fourier transform of \( \mathbf{a} \) evaluated at \(u\).
\end{theorem}

\begin{proof}
Let a \emph{ continuous atomic signal} be a pure complex exponential of the form 
\[
s_u(\gamma) = \mathrm{e}^{-\mathrm{j} 2\pi\gamma u},
\]
with \(\gamma \in \mathbb{R}\) and \(u \in \mathbb{R}\).
The signal is then windowed and sampled to obtain a discrete representation as follows \[s_u(\gamma_n)=\mathrm{e}^{- \mathrm{j}2\pi\gamma u}\,\Pi_B(\gamma-\gamma_0)\,\operatorname{III}_{\Delta \gamma}(\gamma),\]
where the windowing is performed using the rectangular function \[\Pi_B(\gamma) =\begin{cases}
1 & \text{if } |\gamma| < \frac{B}{2} \\[6pt]
0 & \text{elsewhere}.
\end{cases},\] and the sampling is modeled by the Dirac comb \(
\operatorname{III}_{\Delta\gamma}(\gamma) \;\triangleq\; \sum_{n=-\infty}^{\infty} \delta(\gamma - n\Delta\gamma),\) with $\delta(\cdot)$ denoting the Dirac delta distribution. This yields the following discrete set $\gamma_n \in \left\{ \gamma_0 + n\,\Delta\gamma \;\middle|\;
n = -\frac{B}{2}, -\frac{B}{2}+1, \dots, \frac{B}{2} \right\}$.

\noindent The discretized version of $s_u(\gamma)$ viewed as a function of the target parameter $u$ corresponds to the atomic signal $\mathbf{e}(u)$ (which can be a steering vector or a frequency response vector depending on the system).

\noindent Let an \emph{atom} be an arbitrary continuous signal \( a(\gamma) \) discretized in the same manner so it corresponds to the atom $\mathbf{a}$.

\noindent The response of \(\mathbf{a}\) to \( \mathbf{e}(u)\) is
\begin{align*}
r(u)=\langle \mathbf{a}, \mathbf{e}(u) \rangle &= \sum_n a^\ast(\gamma_n)\, s_u(\gamma_n)\\
&\overset{p}{=} \int_\gamma a^\ast(\gamma)\, \mathrm{e}^{-\mathrm{j}2\pi \gamma u}\,\Pi_B(\gamma-\gamma_0)\,\operatorname{III}_{\Delta \gamma}(\gamma)\,\mathrm{d}\gamma\\
&=\operatorname{FT}\{a^\ast(\gamma)\,\Pi_B(\gamma-\gamma_0)\,\operatorname{III}_{\Delta \gamma}(\gamma)\}(u),
\end{align*}
where $p$ indicates the sifting property \cite{Bracewell2000_Ch5} of the Dirac delta, which justifies the equality.
Consequently, this equality is precisely the definition of the Fourier transform of \( \mathbf{a}^\ast\) evaluated at \(u\). 
\end{proof}
This theorem thus allows us to build meta-atoms as follows, in Corollary \ref{corr:metaatoms}.

\begin{corollary}
\label{corr:metaatoms}
Modulating an atom $\mathbf{a}$ by a sinc function yields a response which is the convolution of the response $r(u)=\langle \mathbf{a}, \mathbf{e}(u) \rangle$ with a rectangular function.

\end{corollary}

\begin{proof}
Let $a_i(\gamma)=\mathrm{e}^{-\mathrm{j}2\pi\gamma u_i}$ be an atomic signal, as is standard in sparse recovery methods. Then by applying Theorem \ref{th:metaatoms}, the response of this atom is given by
\begin{align*}
r_i(u)&=\operatorname{FT}\{\mathrm{e}^{\mathrm{j}2\pi\gamma u_i}\,\Pi_B(\gamma-\gamma_0)\,\operatorname{III}_{\Delta \gamma}(\gamma)\}(u)\\&=\delta(u-u_i) \circledast\, B \operatorname{sinc}\,(Bu)\mathrm{e}^{-j2\pi\gamma_0 u}\,\circledast\,\tfrac{1}{\Delta\gamma}\operatorname{III}_{\tfrac{1}{\Delta\gamma}}(u),
\end{align*}
where $\circledast$ denotes convolution. The resulting response is sinc-shaped with a principal lobe centered at $u_i$.

Now, if we modulate $a_i(\gamma)$ such that
\[
\hat{a}_i(\gamma) = a_i(\gamma) \,\operatorname{sinc}\,(L\gamma),
\]
the resulting correlation response becomes
\begin{align*}
\hat{r}_i(u) &= \operatorname{FT}\{\operatorname{sinc}\,(L\gamma)\,\mathrm{e}^{\mathrm{j}2\pi\gamma u_i}\,\Pi_B(\gamma-\gamma_0)\,\operatorname{III}_{\Delta \gamma}(\gamma)\}(u) \\
&=\tfrac{1}{L}\Pi_L(u-u_i) \circledast\, B \operatorname{sinc}\,(Bu)\mathrm{e}^{-\mathrm{j}2\pi\gamma_0 u}\,\circledast\,\tfrac{1}{\Delta\gamma}\operatorname{III}_{\tfrac{1}{\Delta\gamma}}(u)\\ 
&= \tfrac{1}{L}\,\Pi_L(u) \circledast r_i(u),   
\end{align*}
which is a rectangular-shaped response of width $L$ centered at $u_i$.
\end{proof}
Corollary \ref{corr:metaatoms} shows that by modulating the classical atom $\mathbf{a}_i=\mathbf{e}(\tau_i)$ with a well defined sinc function we get a meta-atom having a spread response of width $L$ \[\hat{\mathbf{a}}_i=\mathbf{e}(\tau_i)\odot\operatorname{sinc}(L\mathbf{f}),\] where $\odot$ denotes the element-wise (Hadamard) product ($[\mathbf{v} \odot \mathbf{w}]_j = v_j\, w_j$). The resulting structure constitutes the basis of the proposed method named the \emph{hierarchical search}, which will be explained in the following. 

\subsection{Hierarchical search}
This section explains the hierarchical approach and compares it with the classical one. As in the SISO-OFDM scenario, the classical atoms take the form of an FRV, parameterized with a propagation delay. Chosen this way, each atom's response only covers a small interval around one delay in the domain (the atom $\mathbf{a}_i=\frac{1}{\sqrt{N}}\mathrm{e}^{-\mathrm{j}2\pi\mathbf{f}\tau_i}$ has a non-negligible response only around $\tau_i$). Let $A$ be the number of atoms in the dictionary. The set of these atoms (the entire dictionary) must span the entire delay domain over which the search is conducted; hence, the more atoms there are, the higher the resolution. The correlation response is computed in one step ($A$ correlations are computed), and the atom that maximizes it is chosen. 

The hierarchical approach (summarized in Algorithm \ref{algo:Hsearch}) proposes to compute the correlation response iteratively (several steps), by constructing $n$ \emph{meta-atoms} per step, each meta-atom's correlation response covers a wide range of the delay domain, such that the set of meta-atoms covers the entire delay domain, e.g. for $n=2$, the first step consists of constructing $2$ meta-atoms each covering half of the delay domain. In each step, the meta-atom that maximizes the correlation with the signal will be chosen, consequently eliminating the delay range covered by the other meta-atoms' responses. The same search will then be conducted only in the chosen delay range (covered by the chosen meta-atom's response), and so on. This approach reduces the number of correlations from  $A$ to $n\log_nA$.

\begin{algorithm}[h]
\caption{$\mathtt{HSearch}(\boldsymbol{\epsilon}, \Delta u,\boldsymbol{\gamma}, n, S)$}

\label{algo:Hsearch}

\KwIn{Residual $\boldsymbol{\epsilon}$, target domain length $\Delta u$, observation domain $\boldsymbol{\gamma}$, branching factor $n$, number of steps $S$}

\tcp{Initialize}
$L \gets \Delta u / n$ \tcp*{meta-atom width}
$\mathbf{u} \gets \left( \frac{(2k-1)L}{2} \;\middle|\; k = 1,\dots,n \right)$ \tcp*{meta-atom centers}

\Repeat{$S$ steps}{
    \tcp{Construct meta-atom dictionary}
    $\mathbf{M} \gets \left( \mathbf{e}(u_i)\odot\operatorname{sinc}(L\boldsymbol{\gamma}) \;\middle|\; u_i \in \mathbf{u} \right)$\;

    \tcp{Compute correlation}
    $\mathbf{c} \gets  \mathbf{M}^\mathsf{H} \boldsymbol{\epsilon} $\;

    $j^\star \gets \arg\max_j \; \left|{c}_j\right|$\;
    $u^\star \gets \mathbf{u}[j^\star]$\;

    \tcp{Update}
    $L \gets L / n$\;
    $\mathbf{u} \gets \left( u^\star + \frac{(2k-1-n)}{2}L \mid k = 1,\dots,n \right)$\;
}
\KwOut{Estimate $u^\star$, maximum correlation $\mathbf{c}[j^\star]$}
\end{algorithm}

\begin{figure}[h]
    \centering
    \begin{subfigure}{\columnwidth}
        \centering
        \includegraphics[width=\linewidth]{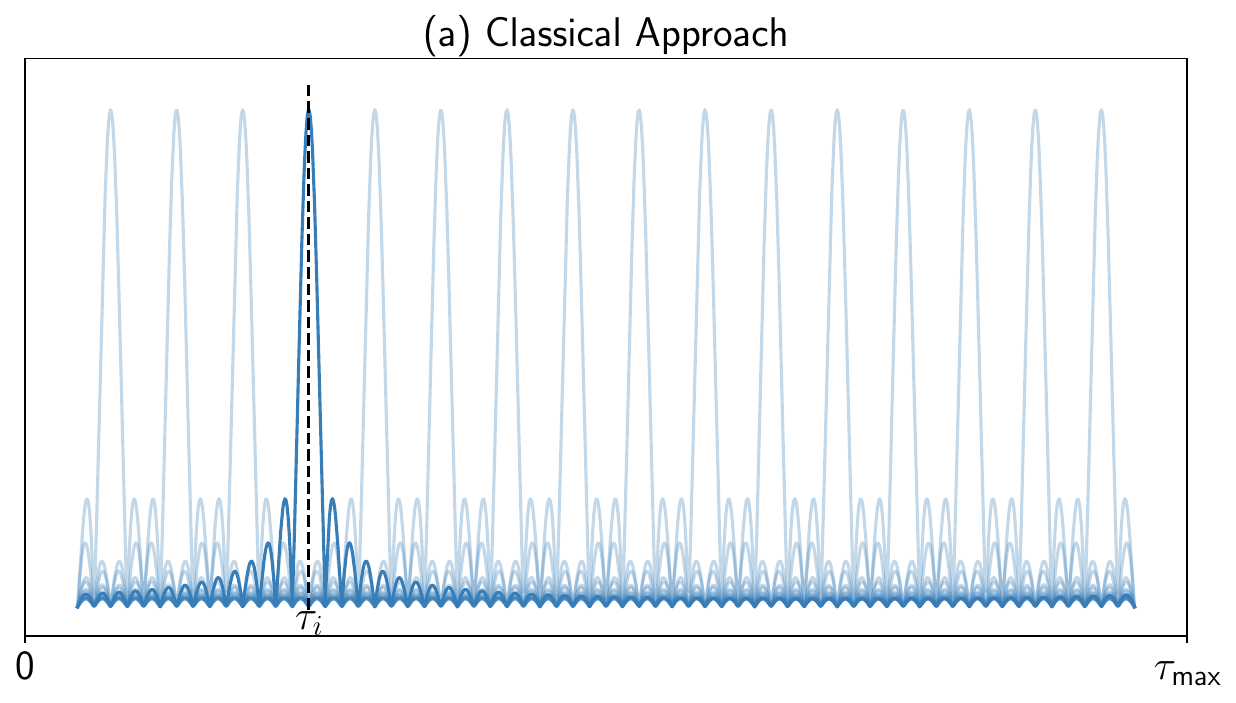}
        \vspace{-1.5em}
    \end{subfigure}
    \vspace{-0.5em}
    \begin{subfigure}{\columnwidth}
        \centering
        \includegraphics[width=\linewidth]{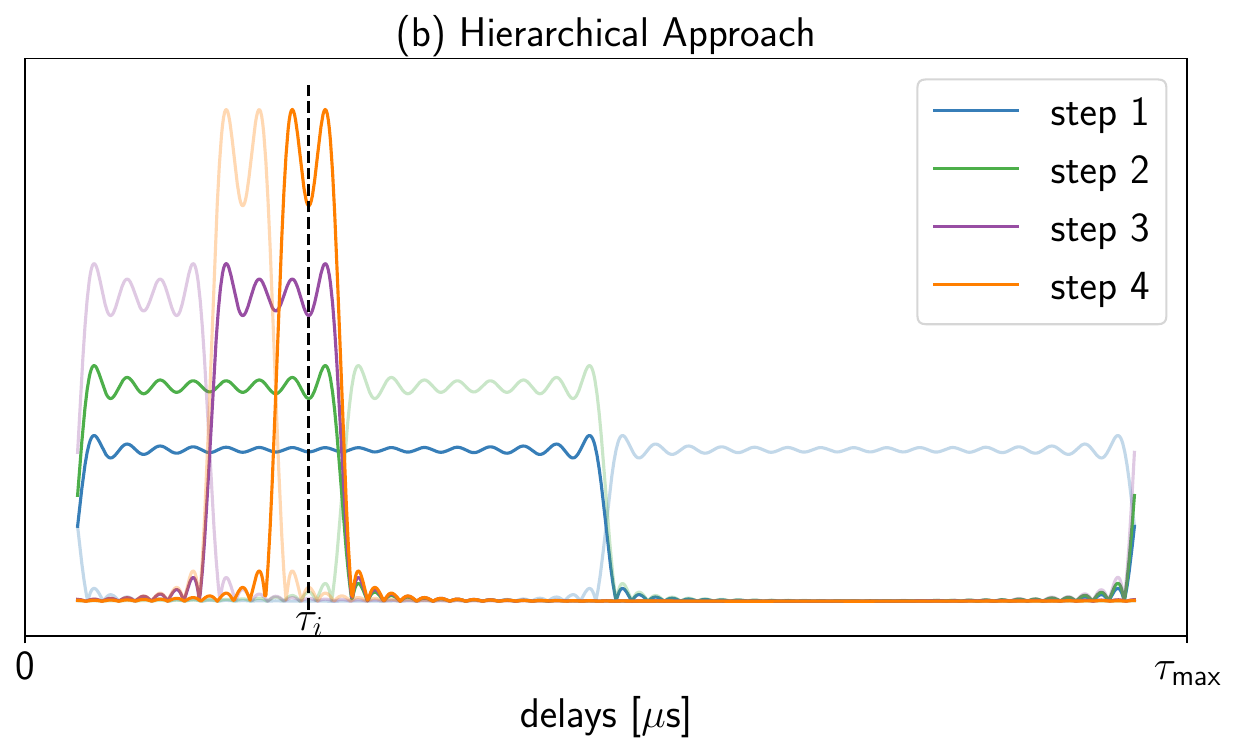}
    \end{subfigure}
    \caption{Classical vs. hierarchical search illustration}

    \label{fig:Approaches}
\end{figure}
The figure \ref{fig:Approaches} illustrates the difference between the two approaches. Let the signal to recover be expressed as $\alpha_i\mathbf{e}(\tau_i)$. For correct recovery, the methods must be able to identify the corresponding delay $\tau_i$. The classical approach (Fig.\ref{fig:Approaches}.a) correlates this signal with the dictionary $\mathbf{D}=\{\mathbf{a_j}\}^{16}_{j=1}$ such that $\mathbf{D}$ spans the entire delay domain ($16$ correlations conducted). The response peaks around $\tau_i$ (line in bold); hence the selected atom is $\mathbf{a}_i$. On the other hand, Fig.\ref{fig:Approaches}.b illustrates the hierarchical approach with $n=2$. The search is conducted within $4$ steps. In each step, the signal is correlated with $2$ meta-atoms and the one having the highest response is maintained (in bold) ($8$ correlations conducted). In the final step, the center of the selected meta-atom corresponds to $\tau_i$.

\subsection{Multi-dimensional Hierarchical OMP}
In this section, we will consider OMP as an example of a sparse recovery algorithm. As for all other algorithms, the hierarchical search can be integrated simply by replacing the classical atom selection with the hierarchical one from Algoritm \ref{algo:Hsearch}. 
For a one-dimensional system, passing from OMP to Hierarchical OMP (HOMP) yields significant complexity gains (see Table \ref{tab:complexity}).
Let us focus now on the multi-dimensional system; a three-dimensional system is considered for illustration and described in section \ref{sec:syst}. The observation is of dimension $\mathbf{y \in \mathbb{C}^{N_B N_M N_S}}$ and the classical dictionary $\mathbf{D}\in \mathbb{C}^{N_B N_M N_S \times A_B A_M A_S}$. Multi-dimensional OMP (MOMP), introduced in \cite{palacios2022,klaimi2026}, allows to treat each dimension of the system independently by viewing the large dictionary $\mathbf{D}$ as the Kronecker product of three smaller ones $\mathbf{D}_B\in \mathbb{C}^{N_B  \times A_B}$, $\mathbf{D}_M\in \mathbb{C}^{N_M  \times A_M}$ and $\mathbf{D}_S\in \mathbb{C}^{N_S \times A_S}$. From this point, atom selection in each dimension is classically performed in MOMP and can be replaced by the proposed hierarchical approach to achieve further complexity gains, particularly as each dimension grows large. The resulting algorithm is called Multi-dimentional Hierarchical OMP (MHOMP) and is summarized in Algorithm \ref{algo:MHOMP}.

\begin{algorithm}[h]
\caption{MHOMP (high level overview)}
\label{algo:MHOMP}
\KwIn{Channel observation $\mathbf{y}$, dictionaries $\{\mathbf{D}_d\}_{d=1}^3$}

Initialize $\boldsymbol{\epsilon} \gets \mathbf{y}$, $\mathbf{D}_\mathsf{active} \gets [\,]$\;

\Repeat{stopping criterion}{
    $u_1, \mathbf{c}_1=\mathtt{HSearch}(\boldsymbol{\epsilon}, \Delta{u_1}, \boldsymbol{\gamma}_1, n, S_1)$ (Algo.\ref{algo:Hsearch})\;
    $u_2, \mathbf{c}_2=\mathtt{HSearch}(\mathbf{c}_1, \Delta{u_2}, \boldsymbol{\gamma}_2, n, S_2)$ \;
    $u_3, c_3=\mathtt{HSearch}(\mathbf{c}_2, \Delta{u_3}, \boldsymbol{\gamma}_3, n, S_3)$ \;
    
    Append atom: $\mathbf{e}_1(u_1) \otimes \mathbf{e}_2(u_2) \otimes \mathbf{e}_3(u_3)$ to $\mathbf{D}_\mathsf{active}$\;
    
    $\mathbf{x}^\star \gets \arg\min_{\mathbf{x}} \|\mathbf{y} - \mathbf{D}_\mathsf{active} \mathbf{x}\|_2$\;
    
    Update the residual: $\boldsymbol{\epsilon} \gets \mathbf{y} - \mathbf{D}_\mathsf{active} \mathbf{x}^\star$\;
}
\KwOut{Denoised channel $\hat{\mathbf{h}} \gets\mathbf{y}-\boldsymbol{\epsilon}$}

\end{algorithm}
Table \ref{tab:complexity} compares the computational complexity of all the discussed atom selection methods, and highlights the relationship between the number of correlations and the number of multiplications, which is standardly used as the theoretical complexity measure. This table demonstrates that hierarchical search is significantly less complex than its classical counterpart, both for one-dimensional systems, and mostly for multidimensional systems when combined with the multidimensional search framework, yielding substantial complexity gains.

\begin{table*}[h]
    \centering
    \begin{tabular}{ccccc}
        \toprule
        & Atom selection method & Number of correlations & Number of multiplications (complexity)  \\
        \midrule
        \multirow{2}{*}{1D}
            & Classical  & $A$ & $\mathcal{O}(NA)$  \\
        \cmidrule(l){2-5}
            & Hierarchical  & $n\log_n(A)$ & $\mathcal{O}\left(Nn\log_n(A)\right)$ \\
        \midrule
        \multirow{4}{*}{3D} 
            & Classical & $A_1A_2A_3$ & $\mathcal{O}(N_1N_2N_3A_1A_2A_3)$ \\
        \cmidrule(l){2-5}
            & Multi-dimentional classical & $A_1+A_2+A_3$ & $\mathcal{O}(A_1N_1N_2N_3+A_2N_2N_3+A_3N_3)$ \\
        \cmidrule(l){2-5}
            & Multi-dimentional hierarchical & $n\log_n(A_1)+n\log_n(A_2)+n\log_n(A_3)$& $\mathcal{O}\left(n\log_n(A_1)N_1N_2N_3+n\log_n(A_2)N_2N_3+n\log_n(A_3)N_3\right)$ \\
        \bottomrule
    \end{tabular}
    \caption{Complexity comparison}\label{tab:complexity}
\end{table*}

\section{Experiments}
{\noindent\bf Settings.} The proposed method is evaluated over two series of experiments: i) a one-dimensional system and ii) a large scale multi-dimensional system using realistic synthetic channels generated with the Sionna ray-tracing simulator \cite{sionna}.
The carrier frequency is of $f = 28\,\text{GHz}$, the subcarrier spacing is set to $120\,\text{kHz}$. A pilot subcarrier is inserted every $12$ subcarriers, resulting in an effective pilot spacing of $\Delta f = 120 \times 12\,\text{kHz} = 1.44\,\text{MHz}$, and a total of $N_S = 256$ pilot subcarriers is considered. 
The first setup comprises a single antenna for both transmitter and receiver, whereas the second comprises a BS equipped with a ULA of $N_B = 64$ antennas, and UE equipped with a ULA of $N_M = 32$ antennas. In the end, this configuration yields channels of dimension $N_B N_M N_S = 524,288$. Finally, $\mathcal{B}=1000$ independent samples are collected for both setups.
To reflect realistic operating conditions, the SNR varies across channel realizations, and an average SNR of 10~dB is considered.
\subsection{One-dimensional setup}
In this setup, the channels are synthesized manually using \eqref{eq:1D_multi-path_channel}, by randomly choosing $\alpha_l \sim \mathcal{CN}(0,1)$ and $\tau_k \sim \mathcal{U}(0,\tau_{\text{max}})$ with $\tau_{\text{max}}=\frac{1}{\Delta{f}}\simeq 7\cdot10^{-7}$ for all samples.
{\noindent\bf User localization.} Fig.\ref{fig:HOMP_delayest} evaluates delay estimation using the classical and hierarchical approaches for single-path channels, as a function of the number of multiplications. The dimension of the system $N_S$ is kept constant, the number of multiplications increases for the classical search as the dictionary size $A$ increases, and for the hierarchical search as the number of steps $S$ increases (see Table~\ref{tab:complexity}). Each point on the purple curve corresponds to a value of $S$, and each point on the orange curve corresponds to a value of $A$, related by $A=2^{S}$. The estimation is evaluated using the mean absolute error $\mathrm{MAE}=\frac{1}{\mathcal{B}}\sum_{j=1}^{\mathcal{B}}\lvert \hat{\tau}_j-\tau_j\rvert$, averaged over all observations. The results show that the hierarchical search achieves near-maximum performance while cutting the number of multiplications by at least a factor of 30.\\
\begin{figure}[h]
    \centering
    \includegraphics[width=0.9\columnwidth]{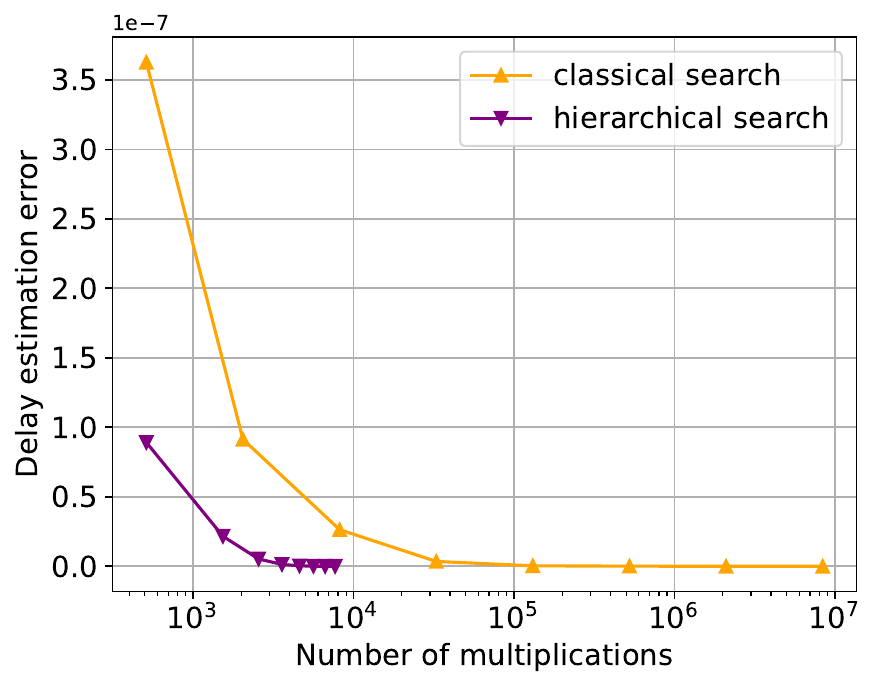}
    \caption{Delay estimation performance for single-path observations at 10~dB SNR for a 1D system}
    \label{fig:HOMP_delayest}
\end{figure}
{\noindent\bf Channel estimation.} The results in Fig.\ref{fig:HOMP_nmse} shows the the normalized mean squared error (NMSE) on the channel estimation obtained with classical OMP algorithm and with HOMP. The channels considered have 3 paths and the NMSE is defined as $\text{NMSE} = \frac{1}{\mathcal{B}}\sum_{j=1}^{\mathcal{B}}{\frac{\|\hat{\mathbf{h}}_j - \mathbf{h}_j\|_2^2}{\|\mathbf{h}_j\|_2^2}}$. It is shown that HOMP achieves channel estimation with more than a 30-fold complexity reduction compared to OMP.\\
\begin{figure}[h]
    \includegraphics[width=0.9\columnwidth]{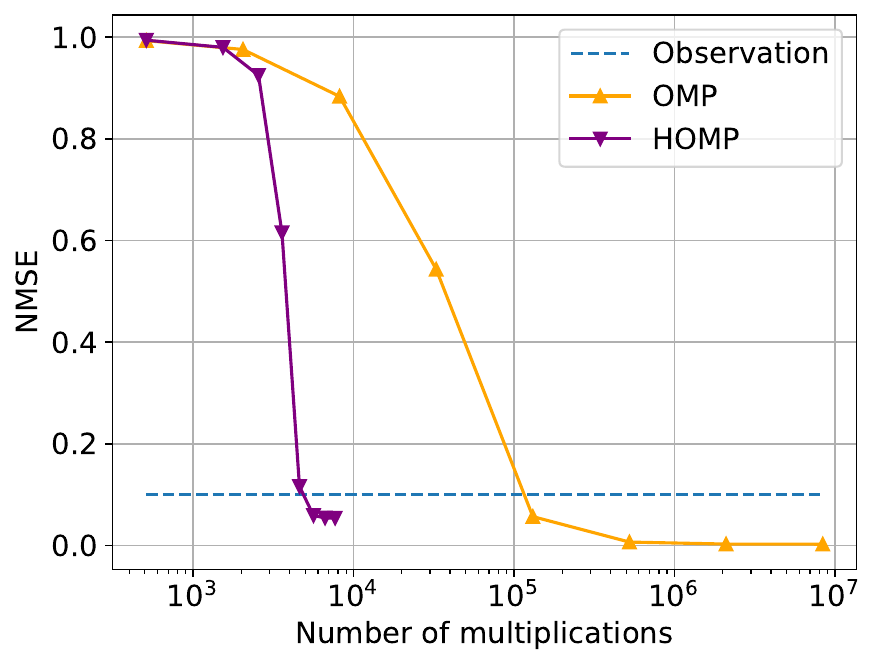}
    \caption{Channel estimation performance on multi-path observations at 10~dB SNR for a 1D system }
    \label{fig:HOMP_nmse}
\end{figure}
Hierarchical search is shown to be effective in achieving massive complexity gains with tolerable performance degradation as the number of paths. However, the scenarios considered so far are of relatively small scale, where the cost of the classical approach remains feasible. The following set of experiences are done on a much larger system in which the classical approach becomes computationally prohibitive.
\subsection{Multi-dimensional setup}
For this setup, realistic synthetic channels are generated using Sionna for the Paris Étoile scenario. The results in Fig.\ref{fig:MHOMP_nmse} show the NMSE on the channel estimation obtained with the methods: classical OMP, MOMP, and MHOMP. The first, being the least suited for large systems, shows very poor performance for fewer than $10^{10}$ multiplications.  Note that, at this system scale, the dictionary required for OMP to achieve satisfactory performance would be of size $524,288\times524,288,000$, corresponding to more than $2\cdot10^{14}$ multiplications, which is computationally prohibitive on standard hardware. This is why the results are shown only up to $10^{10}$ multiplications. MOMP achieves a good complexity reduction with respect to OMP and succeeds in properly estimating the channel at around $10^{10}$  multiplications, which remains a non-negligible cost. MHOMP demonstrates its full potential as each system dimension grows large, achieving almost perfect channel estimation with more than 100-fold complexity reduction.
\begin{figure}[h]
    \includegraphics[width=0.9\columnwidth]{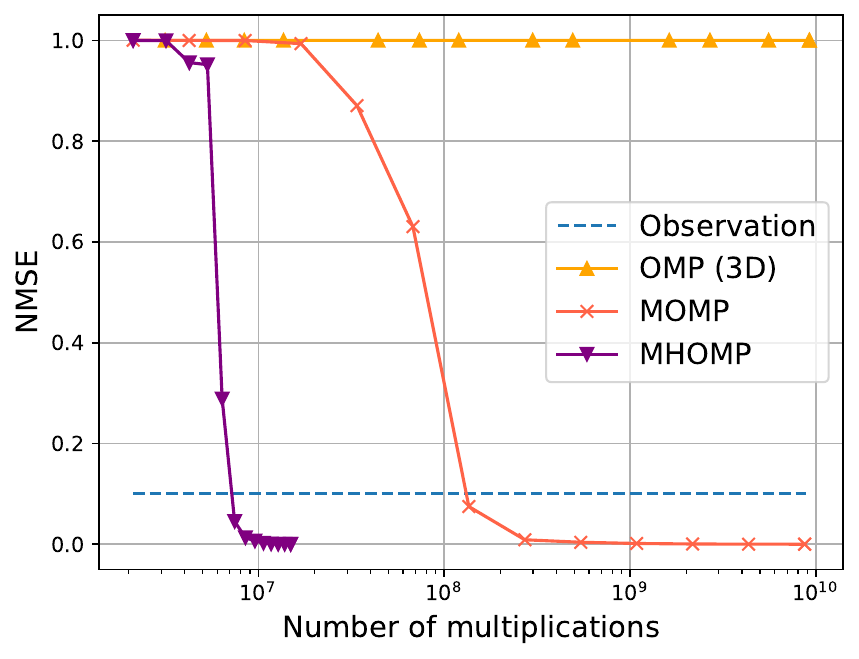}
    \caption{Channel estimation performance on realistic synthetic multi-path observations at 10~dB SNR for a 3D system}
    \label{fig:MHOMP_nmse}
\end{figure}
\section{Conclusion}
This paper introduced a low complexity variant for sparse recovery algorithms, specifically their atom selection step, used for massive MIMO systems. It builds on two key properties of these systems: Fourier-structured atoms, from which a hierarchical atom selection is derived, and Kronecker-structured dictionaries, from which a multidimensional atom selection is derived. The Fourier atomic structure property is formally established and justified. Exploiting jointly the two structures, the proposed hierarchical multidimensional search is shown to achieve a substantial reduction in complexity. The method is then applied to the orthogonal matching pursuit algorithm and evaluated on a large-scale system with realistic synthetic channels. The proposed approach is proven to reduce the number of multiplications by more than two orders of magnitude compared to the classical algorithm for the same level of performance.
\bibliographystyle{IEEEtran} 
\bibliography{references} 

\begin{thebibliography}{10}
\providecommand{\url}[1]{#1}
\csname url@samestyle\endcsname
\providecommand{\newblock}{\relax}
\providecommand{\bibinfo}[2]{#2}
\providecommand{\BIBentrySTDinterwordspacing}{\spaceskip=0pt\relax}
\providecommand{\BIBentryALTinterwordstretchfactor}{4}
\providecommand{\BIBentryALTinterwordspacing}{\spaceskip=\fontdimen2\font plus
\BIBentryALTinterwordstretchfactor\fontdimen3\font minus
  \fontdimen4\font\relax}
\providecommand{\BIBforeignlanguage}[2]{{%
\expandafter\ifx\csname l@#1\endcsname\relax
\typeout{** WARNING: IEEEtran.bst: No hyphenation pattern has been}%
\typeout{** loaded for the language `#1'. Using the pattern for}%
\typeout{** the default language instead.}%
\else
\language=\csname l@#1\endcsname
\fi
#2}}
\providecommand{\BIBdecl}{\relax}
\BIBdecl

\bibitem{Rusek2013}
F.~Rusek, D.~Persson, B.~K. Lau, E.~G. Larsson, T.~L. Marzetta, O.~Edfors, and
  F.~Tufvesson, ``Scaling up {MIMO}: Opportunities and challenges with very
  large arrays,'' \emph{IEEE Signal Processing Magazine}, vol.~30, no.~1, pp.
  40--60, 2013.

\bibitem{Larsson2014}
E.~G. Larsson, O.~Edfors, F.~Tufvesson, and T.~L. Marzetta, ``Massive {MIMO}
  for next generation wireless systems,'' \emph{IEEE Communications Magazine},
  vol.~52, no.~2, pp. 186--195, 2014.

\bibitem{Heath2016}
R.~W. Heath, N.~González-Prelcic, S.~Rangan, W.~Roh, and A.~M. Sayeed, ``An
  overview of signal processing techniques for millimeter wave {MIMO}
  systems,'' \emph{IEEE Journal of Selected Topics in Signal Processing},
  vol.~10, no.~3, pp. 436--453, 2016.

\bibitem{Pati1993}
Y.~Pati, R.~Rezaiifar, and P.~Krishnaprasad, ``Orthogonal matching pursuit:
  recursive function approximation with applications to wavelet
  decomposition,'' in \emph{Proceedings of 27th Asilomar Conference on Signals,
  Systems and Computers}, 1993, pp. 40--44 vol.1.

\bibitem{Chatelier2023}
B.~Chatelier, L.~Le~Magoarou, and G.~Redieteab, ``Efficient deep unfolding for
  {SISO-OFDM} channel estimation,'' in \emph{ICC 2023 - IEEE International
  Conference on Communications}, 2023, pp. 3450--3455.

\bibitem{Jost2006}
P.~Jost, P.~Vandergheynst, and P.~Frossard, ``Tree-based pursuit: Algorithm and
  properties,'' \emph{IEEE Transactions on Signal Processing}, vol.~54, no.~12,
  pp. 4685--4697, 2006.

\bibitem{Ayremlou2014}
\BIBentryALTinterwordspacing
A.~Ayremlou, T.~A. Goldstein, A.~Veeraraghavan, and R.~Baraniuk, ``Fast
  sublinear sparse representation using shallow tree matching pursuit,''
  \emph{ArXiv}, vol. abs/1412.0680, 2014. [Online]. Available:
  \url{https://api.semanticscholar.org/CorpusID:8590846}
\BIBentrySTDinterwordspacing

\bibitem{Skretting2017}
\BIBentryALTinterwordspacing
K.~Skretting and K.~Engan, ``Sparse approximation by matching pursuit using
  shift-invariant dictionary,'' in \emph{Scandinavian Conference on Image
  Analysis}, 2017. [Online]. Available:
  \url{https://api.semanticscholar.org/CorpusID:21897612}
\BIBentrySTDinterwordspacing

\bibitem{Dorffer2019}
C.~Dorffer, C.~Herzet, and A.~Drémeau, ``Region-based relaxations to
  accelerate greedy approaches,'' in \emph{2019 27th European Signal Processing
  Conference (EUSIPCO)}, 2019, pp. 1--5.

\bibitem{palacios2022}
\BIBentryALTinterwordspacing
J.~Palacios, N.~González-Prelcic, and C.~Rusu, ``Multidimensional orthogonal
  matching pursuit: theory and application to high accuracy joint localization
  and communication at {mmWave},'' 2022. [Online]. Available:
  \url{https://arxiv.org/abs/2208.11600}
\BIBentrySTDinterwordspacing

\bibitem{Bayraktar2024}
M.~Bayraktar, N.~González-Prelcic, G.~C. Alexandropoulos, and H.~Chen,
  ``{RIS}-aided joint channel estimation and localization at {mmWave} under
  hardware impairments: A dictionary learning-based approach,'' \emph{IEEE
  Transactions on Wireless Communications}, vol.~23, no.~12, pp.
  19\,696--19\,712, 2024.

\bibitem{klaimi2026}
\BIBentryALTinterwordspacing
N.~Klaimi, C.~Elvira, P.~Mary, and L.~Le~Magoarou, ``Physically constrained
  unfolded multi-dimensional {OMP} for large {MIMO} systems,'' 2026. [Online].
  Available: \url{https://arxiv.org/abs/2601.10771}
\BIBentrySTDinterwordspacing

\bibitem{Mallat1993}
S.~Mallat and Z.~Zhang, ``Matching pursuits with time-frequency dictionaries,''
  \emph{IEEE Transactions on Signal Processing}, vol.~41, no.~12, pp.
  3397--3415, 1993.

\bibitem{Bracewell2000_Ch5}
R.~N. Bracewell, \emph{The Impulse Symbol}, 3rd~ed.\hskip 1em plus 0.5em minus
  0.4em\relax New York: McGraw-Hill, 2000, ch.~5, pp. 74--104.

\bibitem{sionna}
J.~Hoydis, S.~Cammerer, F.~{Ait Aoudia}, M.~Nimier-David, L.~Maggi, G.~Marcus,
  A.~Vem, and A.~Keller, ``Sionna,'' 2022, https://nvlabs.github.io/sionna/.

\end{thebibliography}
\vspace{12pt}

\end{document}